\documentclass[12pt]{article}

\usepackage{amsmath,amssymb,amsthm,bm,xcolor,thmtools,url}

\newtheorem{Theorem}{Theorem}
\newtheorem{lemm}{Lemma}
\newtheorem{Assumption}{Assumption}

\newtheorem{Corollary}[Theorem]{Corollary}

\theoremstyle{remark}
\newtheorem{Remark}{Remark}

\usepackage[T1]{fontenc}
\usepackage[utf8]{inputenc}
\usepackage[round]{natbib}
\usepackage{booktabs}
\usepackage{multirow}
\usepackage{float}
\usepackage{graphicx}

\usepackage{rotating}

\newtheoremstyle{mystyle}%
{3pt}% Space above
{3pt}% Space below 
{\itshape\color{red}}% Body font
{}% Indent amount
{\bfseries\color{blue}}% Theorem head font
{.}% Punctuation after theorem head
{.5em}% Space after theorem head
{}% Theorem head spec (can be left empty, meaning ‘normal’)

\theoremstyle{mystyle}

\newcommand{\convP}{\stackrel{P}{\to}}

\newcommand{\bs}{\boldsymbol{s}}

\begin{document}

\title{Prediction in functional regression with discretely observed and noisy covariates\bigskip}
\author{Siegfried H\"ormann\thanks{Corresponding author.
Email: shoermann@tugraz.at}\\
    Institute of Statistics, Graz University of Technology\\
    and \\
    Fatima Jammoul\hspace{.2cm} \\
    Institute of Statistics, Graz University of Technology}

\providecommand{\keywords}[1]{\textbf{\textit{Keywords:}} #1}

\maketitle
\date{}

\begin{abstract}
In practice functional data are sampled on a discrete set of observation points and often susceptible to noise. We consider in this paper the setting where such data are used as explanatory variables in a regression problem. If the primary goal is prediction, we show that the gain by embedding the problem into a scalar-on-function regression is limited. Instead we impose a factor model on the predictors and suggest regressing the response on an appropriate number of factor scores. This approach is shown to be consistent under mild technical assumptions, numerically efficient and gives good practical performance in both simulations as well as real data settings.
\end{abstract}
\noindent%
\vfill

\newpage

\section{Introduction}

We consider a sample of functional data $(X_t\colon 1\leq t\leq T)$, where each data point $X_t$ corresponds to a curve $(X_t(s)\colon s\in [0,1])$. These curves can be dependent and stationary or iid. As in practical applications full curves are rarely observed, we furthermore assume that these curves are discretely sampled at the same intraday time points $\bs=(s_1,\ldots, s_p)$ with $0\leq s_1<\cdots<s_p\leq 1$. Often these measurements are susceptible to some sort of noise and thus we actually observe
\begin{equation}\label{signoise}
Z_t=(X_t(s_1),\ldots, X_t(s_p))^\prime+(U_{t1},\ldots, U_{tp})^\prime=:\mathcal{X}_t+U_t.
\end{equation}
We will assume that the noise vectors $(U_t)$ and signals $(\mathcal{X}_t)$ are independent. We are interested in the situation where $p$ is large and may grow with the sample size. More detailed assumptions on the setting will be given later.

The overall goal of this paper is the to investigate a \emph{scalar-on-function} linear model with functional covariates $X_t$ and a scalar response $Y_t$. Hence we consider the relation
\begin{equation}\label{fulregmod}
Y_t=\alpha+\int_0^1 \beta(s)X_t(s)ds+\varepsilon_t,
\end{equation}
for some square integrable function $\beta$ and iid errors $(\varepsilon_t)$. 
Many contributions have then focused on how to estimate the slope curve $\beta(s)$ and establish consistency, rates of convergence and the like. Arguably the most common estimation approach is based on functional principal components (FPCs). The basic idea is to regress $Y_t$ on the principal component scores of $X_t$. The resulting coefficients are estimators for the scores of $\beta(s)$, when this function is expanded along the FPCs. For details we refer e.g.\ to \citet{cardotetal:1999} or

 \citet{hallhorowitz:2007}, where also rates of convergence have been established.

A key difficulty in functional regression is that we are facing an ill-posed problem which requires some regularization techniques. Typical approaches are spectral truncation (e.g.\  \citet{hormann:kidzinski:2015}) or Tikhonov regularisation (e.g.\ \citet{ferratyetal:2012}). \citet{panaretos:chakraborty:2019} propose a hybrid version of both approaches.  \citet{yuanetal:2010} provide a very general method using a reproducing kernel Hilbert space (RKHS). Such regularization techniques can be viewed as different forms of smoothing and thus it is no surprise that spline-based estimation is another important approach in functional regression. For example, \citet{cardotetal:2003} used a penalized B-Spline approach reminiscent of ridge regression in order to fit the regression curve. 

Most of the papers in FDA literature assume that the explanatory variables $X_t$ are fully observed curves. In practice, however, we hardly ever measure a process over a continuum, but rather face a setting as in \eqref{signoise}. 
\citet{ramsaysilverman:2005} explore this question from a pragmatic point of view. By smoothing the predictors $X$ as well as the regression function $\beta(s)$ via some set of appropriately chosen basis functions, the problem is cast in a fully functional setup. \citet{lihsing:2007} 
allow the $X_t$'s to be partially observed as well as potentially noisy. Their estimate for $\beta$ is again obtained by first smoothing $X$ and then using the result in what is essentially a penalized least squares approach. 
 \cite{kneipetal:2016} also consider a discrete sampling in a more general regression problem with `points of impact':
\begin{equation}\label{eq:poi}
Y_t=\alpha+\int_0^1 \beta(s)X_t(s)ds+\sum_{r=1}^S\beta_r X_t(\tilde s_r) + \varepsilon_t.
\end{equation}
For improved estimation of this model we refer to \citet{liebletal:2020}. 

In this paper we consider the setting \eqref{signoise}. Instead of exploring an estimator for $\beta(s)$ itself, however, we solely focus on the prediction problem.
Hence our main target is to find a good predictor $$\widehat{Y}_{T+1}=f_T(Z_1,\ldots, Z_T,Z_{T+1},Y_1,\ldots, Y_T).$$ We thus aim to make $\widetilde{Y}_{T+1}-\widehat{Y}_{T+1}$ small, where $\widetilde{Y}_{T+1}=E(Y_{T+1}|X_1,\ldots, X_T)=Y_{T+1}-\varepsilon_{T+1}$.  \citet{caihall:2006} stress the importance of distinguishing between estimating the regression function and prediction. Is the latter of interest, then regularity of $\hat\beta(s)$ may not be the target. They work with fully observed data and derive rates for a fixed non-random regressor $x$. The results can be extended to noisy data if $n/p=O(1)$. \citet{cardotetal:2007} and \citet{crambesetal:2009} consider this problem indirectly. In their setting the predictor functions $X_t$ are observed (potentially with iid noise) on an equidistantly spaced grid. In essence they then study the distance between $\beta$ and $\hat\beta$ in the semi-norm induced by the covariance operator of the $X_t$. This error in turn is closely related to the error when predicting the conditional mean of $Y_{T+1}$ for any new random function $X_{T+1}$ independent of the sample. The papers \citet{caihall:2006} and \citet{crambesetal:2009} are thus probably the closest related to our setup and target and will hence be methods of comparison. 

In our theoretical results we will lay particular focus on avoiding restrictive smoothness conditions on the explanatory variables $X_t(s)$ and on the slope function $\beta(s)$.
 In Section~\ref{s:structure} we introduce our method, based on an underlying factor model structure of the noisy covariates. We present our theoretical findings  in Section~\ref{s:theobounds}. In Section~\ref{s:sim} we consider a comprehensive simulation study, showcasing our method in cases of both smooth and highly irregular slope functions. A real data example is given in Section~\ref{s:realdata}. We conclude in Section~\ref{s:conclusion}. Proofs and technical lemmas are given in the Appendix.

\section{Exploiting the factor model structure}\label{s:structure}

When looking at our regression problem one may wonder if the detour to the functional model \eqref{fulregmod} is necessary. Since we observe, in fact, a multivariate predictor and not a functional one, it does not seem unreasonable to directly impose a linear model of the form
\begin{equation}\label{discregmod}
Y_t=\alpha+\sum_{j=1}^p \beta_j X_t(s_j)+\varepsilon_t,
\end{equation}
and then to explore the problem from a purely multivariate perspective.  If the number of observation points $p$ is fixed, and the sampling design is the same for all observations (which is typically the case for machine recorded data), then this approach in principle is doable. In the case of noisy covariates, one must first find a way to eliminate the noise $U_t$ or at least explore how it will impact the inference. However, here we are interested in the setting where $p$ is diverging with the sample size $T$. In this case the linear regression machinery becomes more delicate, even if the $X_t$'s were observed without noise. In case $p$ diverges faster than $T$ the problem becomes ill-posed and again requires some regularization approach. Another theoretical and also aesthetic issue is that our model will change with increasing $p$. Hence the coefficients $\beta_j$ are actually of the form $\beta_j^{(p)}$. In an asymptotic analysis we may need to specify what the limiting model is, which then naturally brings us back to the functional view in \eqref{fulregmod}. 
The arguably most important argument against exploiting a model of the form of \eqref{discregmod} is the subsequent collinearity issue when estimating the coefficients $\beta_j$.
Let $X^{(j)}:=(X_1(s_j),\ldots, X_T(s_j))^\prime$. If the sampling points $\bs$ are dense and if the curves $X_t(s)$ are smooth, then neighboring columns $X^{(j)}$ will resemble each other closely. On the other hand, the smoothness of the curves, which is commonly imposed in the literature, assures that we can very well approximate the linear span of $X=(X_1(\bs),\ldots X_T(\bs))^\prime$ by a comparably low dimensional subspace. Most common estimation techniques make use of this fact in one way or the other and thus account for the  \emph{the functional nature of the data}.

In this paper we would like to directly exploit the possibility that $X$ can be sufficiently well approximated by a lower dimensional space. Our approach, however,  is not necessarily tied to smoothness. We will explore the \emph{factor space of an approximate factor model} that can be attributed to functional data of the form \eqref{signoise}. We begin by explaining how such a factor model and the respective factor space is obtained. 

We define the mean function $\mu(s)=EX_t(s)$ and the covariance kernel $\Gamma^X(s,s^\prime)=\mathrm{Cov}(X_t(s),X_t(s^\prime))$, respectively. Assuming that  $\Gamma^X$ is continuous, we obtain
by the Karhunen-Lo\`eve expansion that
\begin{equation}\label{e:mercer}
X_t(s)=\mu(s)+\sum_{\ell\geq 1} x_{t\ell}\varphi_\ell(s),
\end{equation}
where  $\varphi_\ell(s)$ are the eigenfunctions of the covariance operator $\Gamma^X$ and $x_{t\ell}=\int_0^1(X_t(s)-\mu(s))\varphi_\ell(s) ds=:\langle X_t-\mu,\varphi_\ell\rangle$. By Mercer's theorem it follows that the eigenfunctions $\varphi_\ell(s)$ are continuous and that convergence in \eqref{e:mercer} is uniform, in the sense
\begin{equation}\label{e:KLuni}
\sup_{s\in [0,1]}E\left|X_t(s)-\mu(s)-\sum_{\ell= 1}^L x_{t\ell}\varphi_\ell(s)\right|^2\to 0,\quad L\to\infty.
\end{equation}
See e.g.\ \citet{bosq:2000} for details. The scores $(x_{t\ell}\colon \ell\geq 1)$ are uncorrelated and  $\mathrm{Var}(x_{t\ell})=\lambda_\ell$, where $\lambda_\ell$ are the eigenvalues of $\Gamma^X$ (in decreasing order). Choose some integer $L\geq 1$ and define the matrix
$$
B(\bs):=(\sqrt{\lambda_1}\varphi_1(\bs),\ldots,\sqrt{\lambda_L}\varphi_L(\bs)).
$$
Moreover, define $f_t=f_{t,L}=(x_{t1}/\sqrt{\lambda_1},\ldots, x_{tL}/\sqrt{\lambda_L})^\prime$. Then we may write 
\begin{equation}\label{e:low}
X_t(\bs)=\mu(\bs)+B(\bs) f_t+ R_t(\bs),
\end{equation}
with $R_t(\bs)=R_{t,L}(\bs)=X_t(\bs)-\mu(\bs)-B(\bs)f_t$ and $\max_jE|R_{t,L}(s_j)|\convP 0$ as $L\to\infty$.  
Hence, $\mu(\bs)+B(\bs)f_t$ provides an explicit form of an $L$-dimensional proxy of $X_t(\bs)$.  

We are going to incorporate these simple observations in the following manner into our theory.
\begin{Assumption}\label{ass:L}
For $L=L(T)$ large enough, 
we assume that $R_{t,L}(\bs)=0$. The dimension parameter $L$ is allowed to diverge with $T\to\infty$. 
\end{Assumption}
From a practical point of view, Assumption~\ref{ass:L} is not a restriction, since in many real examples $R_{t,L}(s)$ converges to zero rapidly, and hence the error is practically negligible if $L$ is chosen large enough. In its spirit it is related to Assumption~(A3) in \cite{crambesetal:2009}, who impose existence of an $L$ dimensional subspace of functions on which $X_t(s)$ can be uniformly sufficiently well approximated on $[0,1]$. Here we require to have this approximation only on $s\in\{s_1,\ldots, s_p\}$, at the price of requiring that the error becomes $0$ when $L$ is large enough. Our assumption is a theoretical trade-off, which in turn allows to substitute smoothness assumptions and many additional complex technical constraints which are needed in related papers for deriving theoretical results.

Under \eqref{signoise} and Assumption~\ref{ass:L} we then have
\begin{equation}\label{e:factmod}
Z_t=\mu(\bs)+B(\bs) f_t+U_t.
\end{equation}
It holds that $\mathrm{Var}(f_t)=I_L$ and by assumption $\mathrm{Cov}(f_t,U_t)=0$, where $I_L$ denotes the identity matrix in $\mathbb{R}^{L\times L}$. Imposing  that $\mathrm{Var}(U_t)$ is a diagonal matrix, we see that \emph{$Z_t$ follows an $L$-factor model.}

\section{Regressing on the factor scores}\label{s:regression}

Let us now assume that we have Model \eqref{signoise} and that Assumption~\ref{ass:L} holds.  
 Irrespective of whether we impose the functional linear model \eqref{fulregmod} or the multivariate linear model \eqref{discregmod}, using the derived representation for $X_t$ we obtain
\begin{equation}\label{e:regfact}
Y_t=a+b^\prime f_t+\varepsilon_t,
\end{equation}
where $b=(b_1,\ldots, b_L)^\prime$ has components 
$$b_\ell=\int_0^1 \beta(s)\sqrt{\lambda_\ell}\varphi_\ell(s)ds\quad\text{or}\quad b_\ell=\sum_{j=1}^p \beta_j\sqrt{\lambda_\ell}\varphi_\ell(s_j),$$
and
$$
a=\int_0^1\mu(s)\beta(s)ds\quad\text{or}\quad a=\sum_{j=1}^p\beta_j\mu(s_j),
$$
depending on whether we work under \eqref{fulregmod}
or under \eqref{discregmod}, respectively. 
Hence in both cases we obtain an ordinary linear model with explanatory variables $f_t$, which are however not observable and need to be estimated. For fully observed data, we may use $\hat f_t=\int_0^1 X_t(s)\hat\varphi_\ell(s)ds/\sqrt{\hat\lambda_\ell}$, where $\hat\lambda_\ell$ and $\hat\varphi_\ell(s)$ denote the empirical eigenvalues and eigenfunctions related to the sample $X_1,\ldots, X_T$. This hence leads to the classical FPC based estimation schemes. In our more realistic setting, however, we don't observe $X_t(s)$ but rather $Z_t$ as in \eqref{signoise}. We will thus estimate $f_t$ as the factor scores in a factor model, i.e.\ we will pursue a purely multivariate scheme instead of a functional one. A specific estimator and its theoretical properties will be discussed in Section~\ref{s:theobounds}. Before we go into technical details we summarize our general estimation scheme.\bigskip

\underline{{\bf Core algorithm:}}
\begin{enumerate}
\item Estimate $\mu(\bs)$ by $\hat\mu(\bs)=\frac{1}{T}(Z_1+\cdots+Z_{T+1})$.
\item Center the data by $\hat\mu(\bs)$.
\item Choose an appropriate order $\hat L$.
\item Compute estimated factor scores $\hat f_t$.
\item Determine $\hat a$ and $\hat b$ via ordinary least squares.
\item Set $\hat Y_{T+1}=\hat a+\hat b'\hat f_{T+1}$.
\end{enumerate}

Some remarks are due.

\begin{Remark}\label{r:1}
The factor scores are not unique and can be rotated by an orthogonal matrix $G\in\mathbb{R}^{L\times L}$ leading to an equivalent model. For the purpose of prediction the orientation of $\hat f_t$ is irrelevant, because the estimator for the slope will be rotated accordingly. But it has to be noted that the estimators $\hat b$ are not directly comparable when the sample size changes. As a consequence of this, the design matrix $\hat F=[\hat f_1,\ldots, \hat f_{T+1}]'$ needs to be recalculated whenever the sample size changes, because estimates for the factor $f_{T+1}$ may not follow the same rotation as initial estimates for $f_1, \ldots f_T$. This implies that when estimating the factor scores and hence generating the design matrix we need to include the new predictor $Z_{T+1}$. 
\end{Remark}

\begin{Remark}\label{r:3}
An important and non-trivial step in this approach is the estimation of the factor scores $\hat f_t$. In the theoretical framework described in the next section, the estimation is based on a PCA factor model approach. However, the algorithm above can be applied to other factor model methods as well, including a maximum likelihood approach (see e.g. \citet{baili2012}) or a mixture of PCA and maximum likelihood (see e.g. \citet{bailiao2016}). 
\end{Remark}

\begin{Remark}\label{r:4}
A delicate tuning choice in our algorithm concerns the parameter $L$, which determines the number of factors. A natural approach is to use a cross-validation procedure. We refer e.g.\ to the work of \citet{owenwang2016} who developed a Bi-Cross-Validation.  \citet{onatski2010} suggests an empirical eigenvalue distribution approach for this problem. In the setting of this paper, we are not specifically interested in the number of factors $L$ that will recover the signal $X_t$ best, but rather the number of factors $L$ that will provide the best predictions. We propose to address this problem by a cross-validation and refer to Section~\ref{s:sim} for more details on this.
\end{Remark}

\section{Theoretical bounds for the prediction error}\label{s:theobounds}

\subsection{Assumptions}
For a simplified presentation we shall assume from now on that all random variables have zero mean.  Fix $T$ for the moment and then define  $Z = (Z_1, \ldots, Z_{T+1})$ and let $\hat{E} = (\hat{e}_1, \ldots, \hat{e}_L)$ be the eigenvectors of $\frac{1}{T+1}Z^\prime Z$ associated with the largest $L$ eigenvalues. We define
$
\hat F=\sqrt{T+1}\hat E,
$
which denotes the proposed estimator for $F = (f_1, \ldots, f_{T+1})^\prime$. We note that $\frac{1}{T+1}\hat F^\prime \hat F = I_L$. We let $Y=(Y_1,\ldots, Y_{T+1})^\prime$ and $Y_{(-)}=(Y_1,\ldots, Y_{T},0)'$ and define the predictor 
$$\widehat Y_{T+1}:=\frac{1}{T+1}\hat f_{T+1}^\prime\hat F^\prime Y_{(-)}.$$
In other words, the estimate $\widehat Y_{T+1}$ results from the linear model in which $Y_{(-)}$ has been regressed on $\hat F$. Our primary theoretical goal in this paper is to bound $\widehat Y_{T+1}-\widetilde Y_{T+1}$, where $\widetilde Y_{T+1} = f_{T+1}^\prime b$ is the optimal but infeasible predictor. In the following we list the assumptions we are going to use in our proofs. Assumptions~\ref{a:noise}--\ref{a:pcs} stem from \citet{hormann:jammoul:2021} and are required to establish the consistent estimation of the factor scores $f_t$.

\begin{Assumption}\label{a:noise}
The noise process $(U_t)$ is i.i.d.\ zero mean and independent of the signals $(X_t)$. The processes $(U_{ti}\colon 1\leq i\leq p)$ are Gaussian with
 absolutely summable auto-covariances:
$$
\sum_{h\in\mathbb{Z}}|\gamma^U(h)|\leq C_U<\infty.
$$
\end{Assumption}

\begin{Assumption}\label{a:signal}
(a) The process $(X_t\colon t\geq 1)$ is $L^4$-$m$-approximable and has zero mean. (b) The curves $X_t=(X_t(s)\colon s\in [0,1])$ define fourth order random processes (i.e.\ $
\sup_{s\in[0,1]} EX_1^4(s)\leq C_X<\infty$)
with a continuous covariance kernel. 
\end{Assumption}

\begin{Assumption}\label{a:pcs}
For the eigenfunctions $\varphi_\ell$ of the covariance operator $\Gamma^X$ it holds that $$\max_{1\leq k,\ell\leq L}\left|\frac{1}{p}\sum_{i=1}^p\varphi_k(s_i)\varphi_\ell(s_i)\right|=O(1)$$ as $T\to\infty$.
\end{Assumption}

\begin{Assumption}\label{a:regrerror}
The variables $(\varepsilon_t)$ are iid, have zero mean and finite variance $\sigma_\varepsilon^2$. They are independent of $(U_t)$ and $(\chi_t)$. 
\end{Assumption}

Gaussian errors could be avoided at the expense of requiring certain moment inequalities for the noise processes $(U_{ti}\colon 1\leq i\leq p)$. Since we don't require independent noise components, the Gaussian setting is convenient, as the dependence is fully described by the autocovariance function. For ease of presentation, we have chosen to remain within this simplified framework. Furthermore, Assumption~\ref{a:signal} shows that we may consider a much broader class of processes in comparison to existing literature. In particular we require no smoothness assumptions on the underlying predictor, aside from a continuous covariance kernel. The notion of $L^4$-m-approximability allows for a very general dependence structure between the functional observations, including functional ARMA or functional GARCH models. Assumption~\ref{a:pcs} is a merely technical assumption. We note that the corresponding sums are proxies for 
$\int_0^1\varphi_k(s)\varphi_\ell(s)ds$, which is either zero (when $k\neq \ell$) or one (when $k=\ell$). 

\subsection{Consistency rates}

Now we are ready to formulate our theoretical results. In the first theorem we assume that the order of the factor model $L$ is fixed. We will then increase the complexity of the problem.

\begin{Theorem}\label{thmLfixknown} Consider the functional regression model \eqref{fulregmod} with sampling scheme \eqref{signoise}.
Let Assumptions~\ref{ass:L}--\ref{a:regrerror} hold, where $L$ is fixed. Then 
\begin{equation}\label{e:eb}
\vert \widehat Y_{T+1}-\widetilde Y_{T+1} \vert = O_P\left( \frac{1}{\sqrt{T}} + \frac{1}{\sqrt{p}} \right)\quad  \text{as $T$ and $p=p(T)\to\infty$.}
\end{equation}
\end{Theorem}
In this setting we hence obtain a parametric rate of convergence, provided that $T/p=O(1)$. Such a rate was also obtained in \cite{caihall:2006} under a quite different setup, invoking a non-random regressor function $x$ with bounded scores $|\langle x,\varphi_k\rangle|\leq C k^{-\gamma}$ and technical assumptions on $\gamma$. As it is pointed out in \cite{crambesetal:2009},  inference on a fixed $x$ cannot, however, be directly compared to rates of convergence of the prediction error for random regressor functions. These authors in turn obtained a non-parametric rate for the mean square prediction error of the order $O(T^{-(2m+2q+1)/(2m+2q+2)})$ where $m$ provides the number of existing derivatives of $\beta$ and where $q$ is related to the decay rate of $\lambda_j$ by assuming that $\sum_{j\geq k} \lambda_j=k^{-2q}$. Formally, their result compares to ours with $q\to\infty$, which would lead again to the same rate as ours.

With a consistent estimator for $L$, we can obtain the same rates as in Theorem~\ref{thmLfixknown}.

\begin{Corollary}\label{corLfixknown} Consider the same setup as in Theorem~\ref{thmLfixknown}, but assume that $L$ is replaced by a consistent estimator $\hat L$, then \eqref{e:eb} holds.
\end{Corollary}

We deliberately do not further concretise estimating the number of factors $L$. This is generally a delicate problem, but in our context, where the focus is on prediction, $L$ can be easily tuned by cross-validation. We stress here that the out-of-sample prediction error is not necessarily minimized when we choose the correct value of $L$. Let us also note tuning the dimension is a problem which is inherent in other approaches as well. E.g.\ the consistency rates for the predictor obtained in \cite{caihall:2006} depend on correctly tuning the truncation parameter in the PCA estimator of $\beta$. The optimal truncation, leading to the obtained rates, depends on knowledge of the spectrum of $\Gamma^X$ and the decay rates of the scores of the predictor $x$ and the slope $\beta$, when these functions are expanded along the eigenfunctions of $\Gamma^X$. In practice, the truncation parameter is also tuned by cross-validation.

\begin{Theorem}\label{thm}
Let Assumptions~\ref{ass:L}--\ref{a:regrerror} hold. Assume that $p=p(T)\to\infty$ and $L=L(T)\to\infty$.  Then  the one-step prediction error is bounded by
$$\vert \widehat Y_{T+1}-\widetilde Y_{T+1} \vert = O_P\left(\left(\frac{p}{\hat\gamma_L}\right)^4 L^{11/2}\frac{1}{\lambda_L^3} \left( \frac{1}{\sqrt{T}} + \frac{1}{\sqrt{p}} \right) \right).$$
\end{Theorem}

\begin{Corollary}\label{cor}
Suppose that the assumptions of Theorem~\ref{thm} hold. Furthermore, assume that for some $\nu>0$ and some $\rho>0$ we have $\lambda_j\geq \rho j^{-\nu}$ and that there is some $\alpha>0$ such that $p/\hat\gamma_L=O_P(L^{\alpha})$. Then  the one-step prediction error is bounded by
$$\vert \widehat Y_{T+1}-\widetilde Y_{T+1} \vert = O_P\left(L^{4\alpha + 3 \nu + 11/2} \left( \frac{1}{\sqrt{T}} + \frac{1}{\sqrt{p}} \right) \right).$$
\end{Corollary}

Polynomial decay rates for the eigenvalues are commonly assumed in related literature. The condition $p/\hat\gamma_L=O_P(L^{\alpha})$ is discussed in \cite{hormann:jammoul:2021} and can be established if we assume equidistant sampling points and the additional assumption $\sup_s E|X_t(s+h)-X_t(s)|^2=O(h)$ as $h\to 0$. The factor $L^{4\alpha + 3 \nu + 11/2}$ may be viewed as a non-parametric convergence rate due to the increase in the dimension of our model. Thus, if $L$ is growing at slow enough polynomial rate we get the convergence in Corollary~\ref{cor}.

\section{Simulation study}\label{s:sim}

In order to demonstrate our approach and draw comparisons to common techniques, we present a comprehensive simulation study. To this end, we adapt a set of real data to serve as predictors for the simulation setting. The dataset \texttt{pm10} consists of bi-hourly measurements of particulate matter PM10 in Graz, Austria, from October 1st 2010 to March 31st 2011. This dataset has 48 observations over the course of 182 days. To control the smoothness of the underlying signal of our predictor, we pre-smooth this dataset using $21$ cubic B-splines. This resulting functional data object is evaluated at $p=48,96,192$ intraday points. Then, we pull a bootstrap sample of $T=100,200,500,1000$ curves. The resulting curves represent the underlying signal $X_t$, for $t=1, \ldots, T$, which have been observed at the equidistant points $s_j \in \lbrack 0,1 \rbrack$ $j=1, \ldots, p$. In a final step, we add iid normal distributed noise $U_{tj}$ to obtain noisy observations $Z_t(s_j) = X_t(s_j) + U_{tj}$. We consider $U_{tj} \sim N(0, \sigma_U^2)$ with two settings $\sigma_U = 2$ and $\sigma_U= 5$. The resulting curves $Z_t$ represent the noisy predictor that we actually observe. One such observation is illustrated in Figure~\ref{fig:pm10}.

\begin{figure}
\includegraphics[width=7cm]{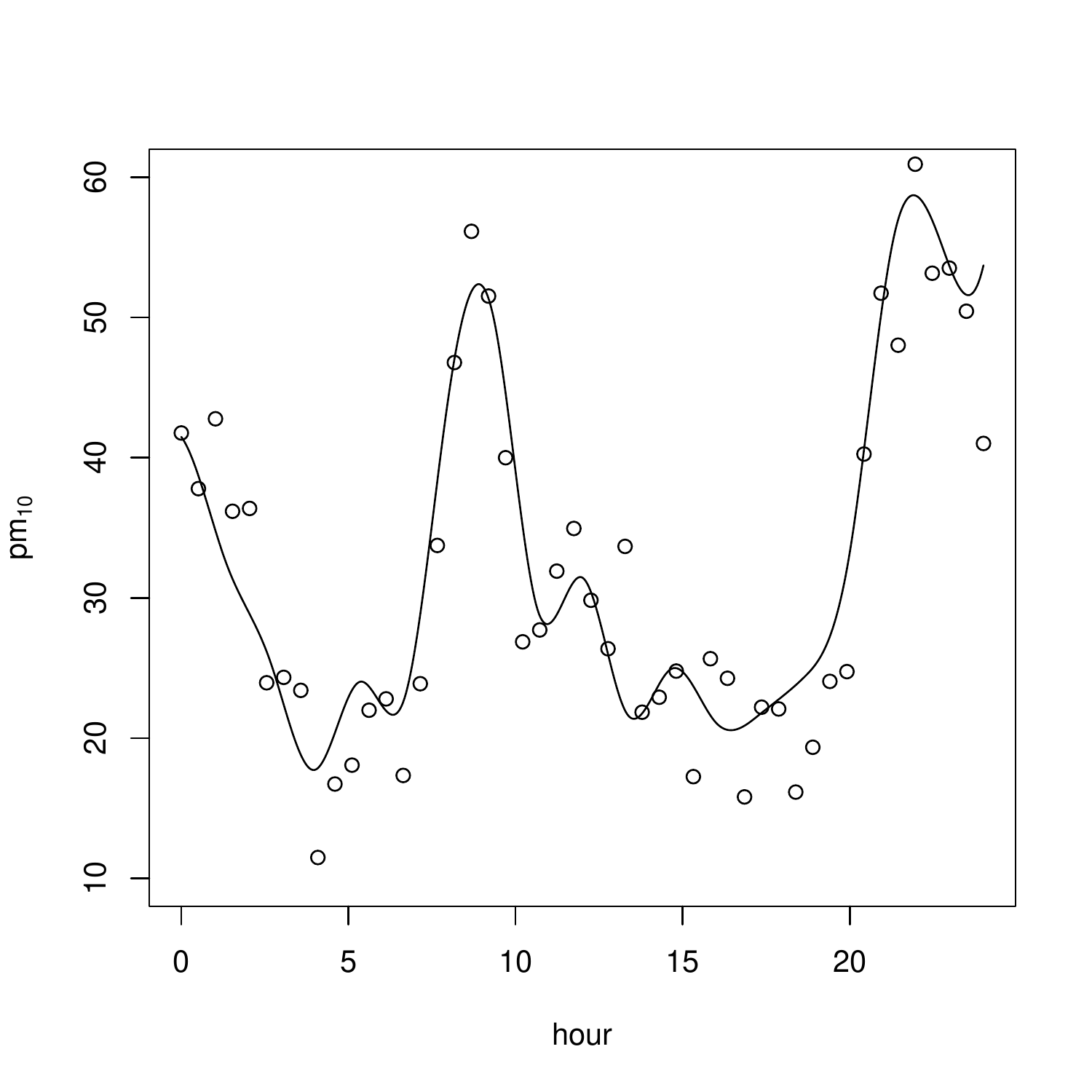}\hfill
\includegraphics[width=7cm]{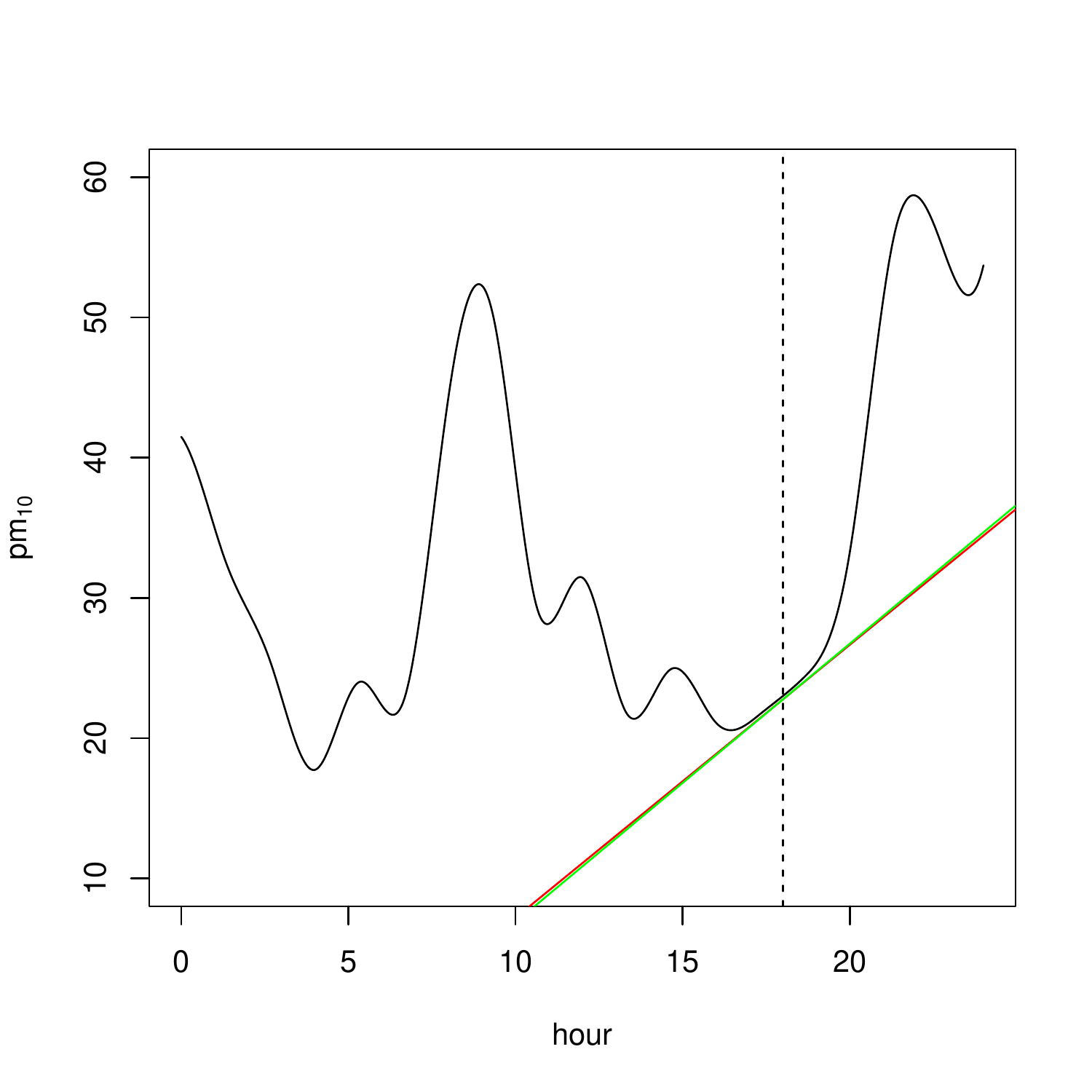}
\caption{Left: Signal (solid line) and signal-plus-noise with $\sigma_U=5$ (dots). Right: Slope of $X_t(s)$ the signal at time $x_0=18$ (green) and approximation of this slope by $\int_0^{24}\beta_{\text{rough}}(s)X_t(s)ds$ (red).}
\label{fig:pm10}
\end{figure}

In the next step, we set up two functional regression models using two very different types of slopes $\beta(s)$. First, we consider the very smooth slope function $\beta_{\text{smooth}}(s) = 10\sin^3{2\pi s^3}$. This particular slope has previously been investigated in \citet{cardotetal:2007}. 
Additionally, we investigate the effects of a very rough slope $$\beta_{\text{rough}}(s):=-1/\epsilon^2 I\{s\in [x_0-\epsilon, x_0)\}+1/\epsilon^2I\{s\in [x_0,x_0+\epsilon]\}.$$
One can see that integration with this slope function is meant to emulate $X_t^\prime(x_0)$ when a small $\epsilon$ is chosen. Since keeping the time scale $s\in [0,1]$ will lead to rather large response values, we have chosen to switch to $s\in [0,24]$. Then $\int_0^{24}\beta_{\text{rough}}(s) X_t(s)ds\approx X_t^\prime(x_0)$ can be interpreted as the rate of change of the {\tt pm10} level at time $x_0$ in $\mu g/(m^3h)$. (See right-hand plot in Figure~\ref{fig:pm10}.) In our simulation we choose $\epsilon=24/100$ and $x_0=18$, which corresponds to the tail end of the evening rush hour. 
In a final step, we obtain the response $Y_t$ for both slopes by setting $$Y_t = \int \beta(s) X_t(s) ds + \varepsilon_t $$ where $\varepsilon_t \sim N(0,\sigma_\varepsilon^2)$ with values $\sigma_\varepsilon \in\{2,5\}$.

Subsequently, we compare three different models: the linear regression model using the recovered factors $f_t$ of $Z_t$ as predictors (\textbf{LM}), the smoothing splines approach for noisy predictors as established in \citet{crambesetal:2009} and \citet{cardotetal:2007} (\textbf{Smooth}) and the functional linear regression with points of impact approach as described in \citet{kneipetal:2016} and combined with its improved estimation in \citet{liebletal:2020} (\textbf{PoI}).  We note that the method \textbf{PoI} is not designed to accommodate noisy observations, and hence the results are not directly compatible with our setup. However, we also experimented with the smoothing splines approach for non-noisy predictors as described in \citet{crambesetal:2009}. These results very closely resembled the \textbf{Smooth} approach and hence it is not unreasonable to assume that \textbf{PoI} might give competitive results.
We also investigated the use of a pre-smoothing step as suggested in \citet{crambesetal:2009}. Here, for more complex noise structures, it was suggested to use a technique to recover the signal from the noisy predictor first and using this estimate to subsequently estimate the functional linear model with  \citet{crambesetal:2009}.
We found, however, that this extra step didn't give any significant improvement in the predictions compared to the  \textbf{Smooth} approach, which is why we will not elaborate on these results further.

It is important to note that for the \textbf{LM} approach   one needs to first estimate the number of factors $L$ required to represent the underlying signal. 

Since in this paper we are interested in prediction, we use generalized cross-validation to obtain an estimate for the number of factors. To this end we choose an upper limit $L_{\text{max}}$ and proceed to fit the linear model with all possible choices for $\ell \leq L_{\text{max}}$. We choose the number of factors that minimizes the GCV-score $$\text{GCV}(\ell) = \frac{T^{-1}\sum_{t=1}^T(Y_t - \widehat Y_{t,\ell})^2}{(1-\ell/T)^2}.$$ In a linear regression this score is a numerically very efficient coefficient which approximates the leave-one-out cross-validation error (see e.g. \citet{hastieetal:2001}). Note that here, the factors have to only be estimated once using the maximum number $L_{\text{max}}$ and may then be added one by one into the linear model to then calculate $\text{GCV}(\ell)$, which helps speed up the implementation especially for large datasets. We have set $L_{\text{max}} = 25$.

For the implementation of \textbf{PoI} we use the package \texttt{FunRegPoI} as provided in the supplementary material to \citet{liebletal:2020}. This method uses a modified version of the approach described in \citet{crambesetal:2009} for the estimation of the $\beta(s)$ and $\beta_k$. 
In the smooth setting the number of points of impact is 0. In this case, the estimation corresponds to the method for non-noisy data as described in \citet{crambesetal:2009}.  
Due to the spiked nature of the slope function $\beta_{\text{rough}}$ we expect that the method might indicate points of impact in this setting. 

The implementation of the method \textbf{PoI} requires a choice of a maximum number of points of impact $S_{\text{max}}$, for which we have found the suggested choice $S_{\text{max}}=8$ to be sufficient. To apply the methods of \citet{crambesetal:2009} and \citet{liebletal:2020}, one must initially choose the order of the smoothing splines estimators. In accordance with both references, we have chosen to use cubic smoothing splines. Furthermore, a smoothing parameter must be estimated in the process, which has been achieved via generalized cross validation in analogy to the references.

In Tables~\ref{tab:simsmooth} and~\ref{tab:simrough} we demonstrate the predictive performances of the different approaches mentioned above. To create these numbers, the dataset was extended by 100 testvalues stemming from the respective model. Then, for the $i$-th simulation run, we define 
$$\text{SSE}^{\text{appr}}_i = \frac{1}{100}\sum_{t=1}^{100}(\widetilde{Y}_{T+t} - \widehat{Y}_{T+t})^2,$$
where $\widehat{Y}_t$ refers to the predicted response. We repeat this simulation for each setting 200 times and report the average of the $\text{SSE}^{\text{appr}}_i$ for $i=1,\ldots, 200$, which we denote by $\text{SSE}^{\text{appr}}$. We also report the median for the estimated number of points of impact $\hat S$ and the median for the chosen number of factors $\hat L$. We remark that the true $L$ in our chosen setup is $21$.

The approaches \textbf{Smooth} and \textbf{LM} appear to have different strengths when we work with the smooth slope function. The \textbf{LM} approach performs best through the setups for large sample sizes and smaller regression errors. In turn, the method \textbf{Smooth} works particularly well when we have a small sample size $T$ and larger regression errors. The \textbf{PoI} approach is less competitive for this data. It erroneously indicates points of impact. We thus conjecture that the method is not robust to noisy observations. In the case of $\beta_{\text{rough}}$ (Table~\ref{tab:simrough}), the \textbf{PoI} can exploit its strength in cases of small $p$ and $T$ but still suffers in other setups. Here \textbf{LM} is generally performing best. The most favourable situation for \textbf{LM} is when $p$ is small and $T$ is large. It is also of note that in the case of the smooth slope, the number of factors chosen by the GCV tends to be smaller than the true number of factors $21$, whereas for the rough slope $\hat{L}$ is close to its actual value.

\begin{table}
\centering
\begingroup\small
\begin{tabular}{ccc|ccccc|ccccc}
  \toprule \multicolumn{3}{c}{Dimensions} & \multicolumn{1}{|c}{} & \multicolumn{4}{c}{$\text{SSE}^{\text{appr}}$ ($\sigma_U$ = 2)} & \multicolumn{1}{|c}{} & \multicolumn{4}{c}{$\text{SSE}^{\text{appr}}$ ($\sigma_U$ = 5)}\\$p$ & $T$ & $\sigma_\varepsilon$ & $\hat{L}$ & $\hat{S}$ & PoI & Smooth & LM & $\hat{L}$ & $\hat{S}$ & PoI & Smooth & LM \\ 
  \hline
48 & 100 & 2 & 15 & 3 & 3.63 & \textbf{3.13} & 3.66 & 13 & 2 & 14.23 & \textbf{13.14} & 13.91 \\ 
  48 & 200 & 2 & 16 & 3 & 2.75 & 2.5 & \textbf{2.46} & 14 & 1 & 12.07 & \textbf{11.73} & 11.87 \\ 
  48 & 500 & 2 & 19 & 3 & 2.22 & 2.17 & \textbf{1.97} & 16 & 2 & 10.96 & 10.94 & \textbf{10.77} \\ 
  48 & 1000 & 2 & 20 & 4 & 2.07 & 2.09 & \textbf{1.84} & 17 & 2 & 10.54 & 10.6 & \textbf{10.37} \\ 
   &  &  &  &  &  &  &  &  &  & & \\ 
96 & 100 & 2 & 15 & 3 & 2.92 & \textbf{2.26} & 2.58 & 14 & 1 & 8.18 & \textbf{7.13} & 7.85 \\ 
  96 & 200 & 2 & 16 & 3 & 1.84 & 1.62 & \textbf{1.49} & 15 & 1 & 6.44 & \textbf{6.03} & 6.2 \\ 
  96 & 500 & 2 & 19 & 3 & 1.39 & 1.36 & \textbf{1.1} & 16 & 1 & 5.79 & 5.55 & \textbf{5.53} \\ 
  96 & 1000 & 2 & 20 & 4 & 1.20 & 1.25 & \textbf{0.94} & 18 & 2 & 5.60 & 5.47 & \textbf{5.34} \\ 
   &  &  &  &  &  &  &  &  &  & & \\ 
192 & 100 & 2 & 15 & 3 & 2.28 & \textbf{1.76} & 1.85 & 15 & 1 & 4.92 & \textbf{3.99} & 4.68 \\ 
  192 & 200 & 2 & 17 & 3 & 1.40 & 1.2 & \textbf{0.98} & 16 & 0 & 3.81 & \textbf{3.43} & 3.53 \\ 
  192 & 500 & 2 & 17 & 4 & 0.94 & 0.94 & \textbf{0.62} & 17 & 1 & 3.28 & 3.1 & \textbf{2.93} \\ 
  192 & 1000 & 2 & 19 & 4 & 0.81 & 0.88 & \textbf{0.53} & 18 & 2 & 3.07 & 2.97 & \textbf{2.71} \\ 
   \midrule
48 & 100 & 5 & 13 & 3 & 9.64 & \textbf{6.98} & 9.9 & 12 & 1 & 20.13 & \textbf{16.8} & 19.25 \\ 
  48 & 200 & 5 & 14 & 3 & 5.15 & \textbf{4.02} & 4.78 & 13 & 1 & 14.70 & \textbf{13.68} & 14.43 \\ 
  48 & 500 & 5 & 15 & 3 & 3.19 & \textbf{2.73} & 2.95 & 14 & 1 & 11.76 & 11.53 & \textbf{11.52} \\ 
  48 & 1000 & 5 & 16 & 3 & 2.50 & \textbf{2.29} & 2.3 & 15 & 1 & 11.16 & 11.03 & \textbf{10.93} \\ 
   &  &  &  &  &  &  &  &  &  & & \\ 
96 & 100 & 5 & 13 & 2 & 9.48 & \textbf{5.95} & 8.67 & 13 & 1 & 14.41 & \textbf{10.7} & 13.86 \\ 
  96 & 200 & 5 & 14 & 3 & 4.70 & \textbf{3.16} & 3.97 & 14 & 1 & 9.10 & \textbf{7.79} & 8.63 \\ 
  96 & 500 & 5 & 15 & 2 & 2.23 & \textbf{1.83} & 1.98 & 15 & 0 & 6.73 & \textbf{6.25} & 6.43 \\ 
  96 & 1000 & 5 & 16 & 3 & 1.76 & 1.53 & \textbf{1.39} & 16 & 0 & 6.12 & \textbf{5.88} & 5.88 \\ 
   &  &  &  &  &  &  &  &  &  & & \\ 
192 & 100 & 5 & 13 & 3 & 9.05 & \textbf{5.15} & 7.66 & 13 & 1 & 11.33 & \textbf{7.83} & 10.74 \\ 
  192 & 200 & 5 & 14 & 3 & 4.43 & \textbf{2.76} & 3.53 & 14 & 0 & 6.53 & \textbf{5.07} & 5.87 \\ 
  192 & 500 & 5 & 15 & 3 & 1.97 & \textbf{1.48} & 1.55 & 15 & 0 & 4.06 & \textbf{3.66} & 3.78 \\ 
  192 & 1000 & 5 & 16 & 3 & 1.32 & 1.13 & \textbf{0.96} & 16 & 0 & 3.50 & 3.24 & \textbf{3.2} \\ 
   \bottomrule
  \end{tabular}
  \endgroup
\caption{Simulation Results for the synthetic \texttt{PM10} data with \textit{smooth} slope function $\beta_{\text{smooth}}(s)$.} 
\label{tab:simsmooth}
\end{table}

\begin{table}
\centering
\begingroup\small
\begin{tabular}{ccc|ccccc|ccccc}
  \toprule \multicolumn{3}{c}{Dimensions} & \multicolumn{1}{|c}{} & \multicolumn{4}{c}{$\text{SSE}^{\text{appr}}$ ($\sigma_U$ = 2)} & \multicolumn{1}{|c}{} & \multicolumn{4}{c}{$\text{SSE}^{\text{appr}}$ ($\sigma_U$ = 5)}\\$p$ & $T$ & $\sigma_\varepsilon$ & $\hat{L}$ & $\hat{S}$ & PoI & Smooth & LM & $\hat{L}$ & $\hat{S}$ & PoI & Smooth & LM \\ 
  \hline
48 & 100 & 2 & 21 & 5 & 7.89 & 7.81 & \textbf{6.97} & 19 & 4 & \textbf{24.85} & 35.67 & 28.92 \\ 
  48 & 200 & 2 & 21 & 6 & 6.72 & 5.66 & \textbf{5.29} & 21 & 5 & \textbf{22.04} & 24.89 & 22.15 \\ 
  48 & 500 & 2 & 21 & 7 & 5.72 & 4.51 & \textbf{4.2} & 21 & 5 & 20.75 & 20.68 & \textbf{18.97} \\ 
  48 & 1000 & 2 & 21 & 7 & 5.36 & 4.23 & \textbf{3.98} & 21 & 6 & 19.63 & 19.85 & \textbf{18.27} \\ 
   &  &  &  &  &  &  &  &  &  & & \\ 
96 & 100 & 2 & 21 & 5 & 7.87 & 4.71 & \textbf{4.64} & 20 & 4 & 26.09 & 19.56 & \textbf{18.68} \\ 
  96 & 200 & 2 & 21 & 6 & 6.43 & 2.99 & \textbf{2.94} & 21 & 5 & 22.72 & 13.81 & \textbf{13.61} \\ 
  96 & 500 & 2 & 21 & 7 & 4.96 & 2.35 & \textbf{2.3} & 21 & 6 & 20.9 & 11.7 & \textbf{11.49} \\ 
  96 & 1000 & 2 & 21 & 8 & 4.21 & 2.17 & \textbf{2.12} & 21 & 7 & 18.72 & 10.93 & \textbf{10.54} \\ 
   &  &  &  &  &  &  &  &  &  & & \\ 
192 & 100 & 2 & 21 & 6 & 7.85 & \textbf{3.03} & 3.09 & 20 & 4 & 25.83 & \textbf{10.19} & 10.74 \\ 
  192 & 200 & 2 & 21 & 6 & 6.26 & \textbf{1.81} & 1.83 & 20 & 6 & 21.49 & \textbf{7.56} & 7.8 \\ 
  192 & 500 & 2 & 21 & 7 & 4.96 & 1.31 & \textbf{1.29} & 21 & 7 & 18.75 & \textbf{6.29} & 6.38 \\ 
  192 & 1000 & 2 & 21 & 8 & 4.32 & 1.16 & \textbf{1.13} & 21 & 7 & 16.72 & 6.09 & \textbf{6.03} \\ 
   \midrule
48 & 100 & 5 & 21 & 4 & \textbf{12.27} & 17.29 & 17.12 & 18 & 4 & \textbf{30.52} & 43.5 & 37.43 \\ 
  48 & 200 & 5 & 21 & 4 & 8.61 & 8.69 & \textbf{8.22} & 20 & 4 & \textbf{23.87} & 28.64 & 25.96 \\ 
  48 & 500 & 5 & 21 & 6 & 6.91 & 5.75 & \textbf{5.42} & 21 & 5 & 21.7 & 22.18 & \textbf{20.48} \\ 
  48 & 1000 & 5 & 21 & 7 & 5.88 & 4.75 & \textbf{4.47} & 21 & 6 & 20.52 & 20.63 & \textbf{18.88} \\ 
   &  &  &  &  &  &  &  &  &  & & \\ 
96 & 100 & 5 & 20 & 3 & \textbf{11.5} & 11.94 & 12.97 & 19 & 4 & 30.12 & 28.3 & \textbf{27.14} \\ 
  96 & 200 & 5 & 21 & 4 & 8.52 & \textbf{5.95} & 6.09 & 20 & 5 & 25.38 & \textbf{16.88} & 16.89 \\ 
  96 & 500 & 5 & 21 & 6 & 6.18 & 3.45 & \textbf{3.44} & 21 & 5 & 21.57 & 12.76 & \textbf{12.5} \\ 
  96 & 1000 & 5 & 21 & 7 & 4.94 & 2.63 & \textbf{2.59} & 21 & 6 & 20.15 & 11.65 & \textbf{11.33} \\ 
   &  &  &  &  &  &  &  &  &  & & \\ 
192 & 100 & 5 & 19 & 3 & 11.36 & \textbf{10.15} & 11.06 & 19 & 4 & 30.84 & \textbf{18.31} & 18.79 \\ 
  192 & 200 & 5 & 20 & 4 & 8.63 & \textbf{4.76} & 5.04 & 20 & 5 & 24.24 & \textbf{10.47} & 10.75 \\ 
  192 & 500 & 5 & 21 & 6 & 5.99 & \textbf{2.37} & 2.4 & 20 & 6 & 19.67 & \textbf{7.29} & 7.33 \\ 
  192 & 1000 & 5 & 21 & 7 & 5.14 & \textbf{1.65} & \textbf{1.65} & 21 & 7 & 17.85 & 6.68 & \textbf{6.61} \\ 
   \bottomrule
  \end{tabular}
\endgroup
\caption{Simulation Results for the synthetic \texttt{PM10} data with \textit{rough} slope function $\beta_{\text{rough}}(s)$.} 
\label{tab:simrough}
\end{table}

\section{Data illustration}\label{s:realdata}
We consider weather-related data recorded in Canadian Weather stations in the neighbouring provinces of Quebec and Ontario. Here, daily mean temperature measurements in the year of 2013 are considered. These yearly curves will serve as our predictors, which we will use to model cumulative log-precipitation in the same year for each station. This type of model has been considered before in \citet{ramsaysilverman:2005}, but we use an extended dataset available at \url{https://climate.weather.gc.ca/} including more curves. 

In an initial data-clean-up phase we eliminate weather stations  for which more than $15 \%$ of the temperature data is missing. For the remaining missing values, we simply impute them by taking the mean of two adjacent values. Is no such mean available, we simply repeat the last observation. For the associated response, the same removal and imputation methods were used. On occasion, multiple stations with the same name and same geographical position were found in the data. These copies were removed from considerations, as it was difficult to ascertain their nature. After these steps, we have $T=193$ weather stations with $p=365$ observations for each predictor curve. A selection of which can be seen in Figure~\ref{fig:canadaweather}. We observe the typical seasonal shape associated with temperature data. The data is then separated randomly into a trainingset of size $T_1 = 143$ and a testset of size $T_2 = 50$. For our analysis the response as well as the predictors have been centered. In a first step, we compute the estimated number of factors using the generalized cross-validation technique. The results of \textbf{LM} are subsequently compared to the methods \textbf{Smooth} and \textbf{PoI}.  This split into training- and test set is then repeated $30$ times with randomly chosen sets. For the mean sum of squared out-of-sample errors and corresponding standard deviation we obtained 0.0393 (0.0073) for \textbf{Smooth}, 0.033 (0.0083) for \textbf{PoI} and 0.0307 (0.0084) for \textbf{LM}. In comparison, the empirical variance of the log-precipitation is $0.0457$ for the total data set. Hence, roughly we can say that our predictions explain approximately $33 \%$ (\textbf{LM}), $27 \%$ (\textbf{PoI}) and $14 \%$ (\textbf{Smooth}) of the variance. We have additionally investigated the difference in the estimated squared out of sample errors in each of the $30$ instances using two-sided paired t-tests. The respective $p$-values are $p=0.2817$ (comparing \textbf{LM} and \textbf{PoI}) $p=8.608 \times 10^{-5}$ (comparing \textbf{LM} and \textbf{Smooth}). Obviously, the errors between the test sets are not independent and hence most likely the actual differences are less significant. So we cannot conclude to have significantly better predictions from \textbf{LM} compared to \textbf{PoI}. However,  the marked difference between \textbf{LM} and \textbf{Smooth} presumably cannot be solely attributed to this fact. 

To conclude this prediction exercise, we remark that the estimated number of factors $\hat L$ was mostly between $26-33$. %, as seen in Figure~\ref{fig:canadaboxplots}. 
The PoI approach estimated between 3 and 5 points of impact, where 2 April (15 times), 25 September (10 times) and 25 May (8 times) were chosen most often.

\begin{figure}
\centering
\includegraphics[width=12cm]{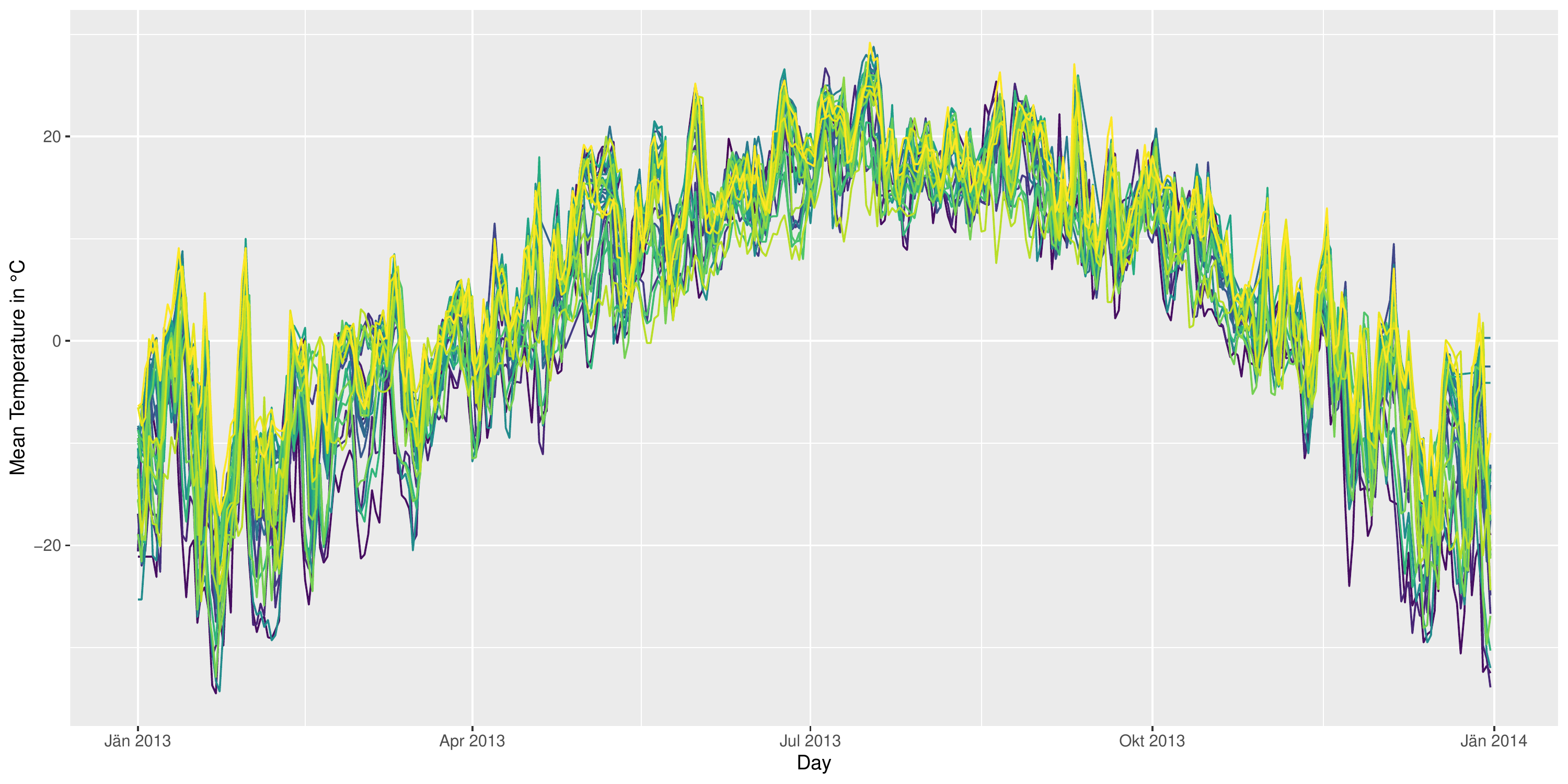}
\caption{Mean Temperature curves in $2013$ for $30$ out of $193$ available Canadian Weather Stations in the Province of Quebec and Ontario }\label{fig:canadaweather}
\end{figure}

\section{Conclusion}\label{s:conclusion}

We have been considering scalar-on-function regression with discretely observed and noisy predictors. Although we impose a popular functional model underlying our data, we can tackle this problem very well from a purely multivariate perspective; this is our key message. The approach we propose is rooted in factor model analysis. It does not require smoothness assumptions on the functional predictors, can accommodate dependent covariates and also allows for dependent sampling errors. We obtain error rates to the theoretically optimal prediction which are comparable to existing results, but work under fewer and simpler assumptions. From a practical side, we show in simulations and real data that our approach gives very convincing results in cases of both smooth and very irregular slope functions. An additional asset is its simple implementation and its fast running time.
\appendix
\section{Proofs}
We will prove Theorem~\ref{thm}. The proofs of Theorem~\ref{thmLfixknown}, Corollary~\ref{corLfixknown} and Corrollary~\ref{cor} follow easily from this result.
\subsection{Decomposing the prediction error}

To expand the prediction error we require further notation. Let use write $B=B(\bs)$ and introduce
$$
H = \frac{1}{T+1}\hat{\Lambda}^{-1} \hat{F}^\prime F B'B,
$$
where $\hat{\Lambda}=\mathrm{diag}(\hat\gamma_1,\ldots,\hat\gamma_L)$ with $\hat\gamma_1\geq \cdots\geq \hat\gamma_L$ being the $L$ largest empirical eigenvalues of $\frac{1}{T+1}Z'Z$. The matrix $H$ takes the role of $G$ in Remark~\ref{r:1}. It can be shown to be asymptotically orthogonal and is used to fix an orientation. We write the regression model in the form \eqref{e:regfact} and then obtain 
$$
\widetilde{Y}_{T+1} = f_{T+1}^\prime H^\prime H b + f_{T+1}^\prime(I_L - H'H)b.
$$
Next, let us write \eqref{e:regfact} in vector notation as
$$
Y=Fb+\varepsilon,
$$
where $\varepsilon=(\varepsilon_1,\ldots,\varepsilon_{T+1})^\prime$. From this we obtain the equation
$$
 b= \frac{1}{T+1}F^\prime Y  - \left(\frac{1}{T+1} F^\prime F  - I_L\right)b - \frac{1}{T+1}F^\prime \varepsilon.
$$

Now we expand the distance between our predictor $\widehat{Y}_{T+1}$ and the best possible predictor $\widetilde{Y}_{T+1}$ in a number of terms which we then shall   bound one by one:
\begin{align*}
\widehat{Y}_{T+1} - \widetilde{Y}_{T+1} &= \frac{1}{T+1}\hat{f}_{T+1}^\prime\hat{F}^\prime Y_{(-)} - f_{T+1}^\prime H^\prime H b - f_{T+1}^\prime(I_L - H^\prime H)b \\
&=  \frac{1}{T+1}\hat{f}_{T+1}^\prime\hat{F}^\prime Y_{(-)} -  \frac{1}{T+1}f_{T+1}^\prime H^\prime H F^\prime  Y \\
&\quad +  f_{T+1}^\prime H^\prime H\left(\frac{1}{T+1} F^\prime F - I_L\right)b  + \frac{1}{T+1}f_{T+1}^\prime H^\prime H F^\prime \varepsilon \\
&\quad +  f_{T+1}^\prime (H^\prime H-I_L )b \\
&=: A + B + C,
\end{align*}
where $A$, $B$ and $C$ correspond to the respective lines of the right hand side of the equation above.

We further expand
\begin{align*}
A &= \frac{1}{T+1}\hat{f}_{T+1}^\prime \hat{F}^\prime Y_{(-)} - \frac{1}{T+1}f_{T+1}^\prime H^\prime HF^\prime Y \\
&= (\hat{f}_{T+1}^\prime - f_{T+1}^\prime H^\prime)\frac{1}{T+1}\hat{F}^\prime Y_{(-)}+ \frac{1}{T+1}f_{T+1}^\prime H^\prime (\hat{F}^\prime Y_{(-)}- HF^\prime  Y ) \\
%&= (\hat{f}_{T+1}' - f_{T+1}'H')\frac{1}{T}\hat{F}_{T+1}(\mathbf{Y}_T',0)' + \frac{1}{T}f_{T+1}'H'\lbrack (\hat{F}_{T+1}' - H'F_{T+1}')\mathbf{Y}_{T+1} + \hat{F}_{T+1}((\mathbf{Y}_T',0)' - \mathbf{Y}_{T+1})\rbrack \\
&= (\hat{f}_{T+1}^\prime - f_{T+1}^\prime H^\prime)\frac{1}{T+1}\hat{F}^\prime Y_{(-)}\\
&\quad + \frac{1}{T+1}f_{T+1}^\prime H^\prime  (\hat{F}^\prime  - H^\prime F^\prime) Y\\
&\quad -  \frac{1}{T+1}f_{T+1}^\prime H^\prime \hat{f}_{T+1}Y_{T+1} \\
&:= A_1 + A_2 - A_3,
\end{align*}
where $A_1$, $A_2$ and $A_3$ correspond to the respective lines of the right hand side of the equation above. Finally, we have
\begin{align*}
B &= f_{T+1} H^\prime H \left(\frac{1}{T+1} F^\prime F  - I_L\right)b + \frac{1}{T+1} f_{T+1} H^\prime H F^\prime \varepsilon := B_1 + B_2.
\end{align*}

Thus, it remains to investigate $A_i$ and $B_i$ for $i=1,2,3$ and $C$. These terms are bounded in a series of lemmas in Appendix \ref{s:proofs} and then finally may be combined to obtain the rate given in Theorem~\ref{thm}.

\subsection{Technical lemmas}\label{s:technical}

For the reader's convenience, we first summarize some results that are used in the proof of Theorem~\ref{thm}. Lemmas~1--3 constitute slight modifications of results found in \citet{hormann:jammoul:2021}. In the following, $\|\cdot\|$ denotes the Euclidian norm of a vector or the usual matrix norm, and $\|\cdot\|_F$ is the Frobenius norm.

\begin{lemm}
Under Assumptions~\ref{ass:L}--\ref{a:signal} we have
\begin{equation}\label{L.4}
  \Vert \hat f_{T+1} - Hf_{T+1} \Vert = O_P\left(\frac{p}{\hat\gamma_L} \sqrt{L}\left(\frac{1}{\sqrt{T}} + \frac{1}{\sqrt{p}}\right)\right)
\end{equation}
and 
\begin{equation}\label{L.12}
  \frac{1}{T+1}\left\Vert \hat{F} - H^\prime F^\prime\right\Vert_F^2 = O_P \left( L^2 \left( \frac{p}{\hat{\gamma}_L}\right)^{2}  \left(\frac{1}{T} + \frac{1}{p} \right) \right).
\end{equation}
\end{lemm}
\begin{proof}
Statement~\eqref{L.12} follows immediately from Lemma~12 in  \citet{hormann:jammoul:2021}. The first statement is a modification of their Lemma~4, in which a uniform bound for $\Vert \hat f_{t} - Hf_{t} \Vert$ over $t$ is derived. Since here we only need a specific value of $t$, we get the smaller error term. The proof of \eqref{L.4} requires to eliminate some additional factors, which have been used in the proof of Lemma~4  in \citet{hormann:jammoul:2021} for bounding the maximum over $t\in \lbrace 1,\ldots, T+1\rbrace$. This modification actually comes with a similar, but simpler proof and hence the detailed derivations are left to the reader.
\end{proof}

\begin{lemm}
Under Assumptions~\ref{ass:L}--\ref{a:pcs} we have
\begin{equation}\label{L.7}
  \Vert H \Vert = O_P\left(\frac{p}{\hat\gamma_L} \frac{L}{\sqrt{\lambda_L}}\right)
\end{equation}
and
\begin{equation}\label{L.8}
  \Vert H'H - I_L \Vert = O_P\left(\left(\frac{p}{\hat\gamma_L}\right)^4 \frac{L^5}{\lambda_L^3} \left(\frac{1}{\sqrt{T}} + \frac{1}{\sqrt{p}}\right) \right).
\end{equation}
\end{lemm}
\begin{proof}
This follows immediately from Lemmas~7 and 8 in \citet{hormann:jammoul:2021}.
\end{proof}

\begin{lemm}
Under Assumption~\ref{a:signal} (a) we have
\begin{equation}\label{L.11}
  \Vert (T+1)^{-1}F^\prime F - I_L \Vert_F = O_P \Big( \frac{L}{\lambda_L \sqrt{T}} \Big).
\end{equation}
\end{lemm}
\begin{proof}
This is Lemma~11 in \citet{hormann:jammoul:2021}.
\end{proof}

\begin{lemm}\label{l:bbound}
We have $\|b\|\leq\sqrt{\lambda_1} \|\beta\|$. 
\end{lemm}
\begin{proof}
We first remark that since $\beta$ is square integrable, we have by Parseval's identity that
$
\int_0^1\beta^2(s)ds=\sum_{k\geq 1} \langle \beta,\varphi_k\rangle^2<\infty.
$
Then $\|b\|^2=\sum_{\ell=1}^L b_\ell^2\leq \sum_{\ell\geq 1} \lambda_\ell\langle \beta,\varphi_\ell\rangle^2\leq \lambda_1\|\beta\|^2$.
\end{proof}

\begin{lemm}\label{l:Ybound}
Under Assumptions~\ref{a:noise} and \ref{a:regrerror} we have $\|Y\|=O(\sqrt{T})$.
\end{lemm}
\begin{proof}
The assumptions imply that $(Y_t)$ defines a stationary and ergodic sequence with 2nd moments.
\end{proof}

\section{Proofs}\label{s:proofs}

\begin{lemm}\label{lemma:As}  Under Assumptions \ref{ass:L}--\ref{a:pcs} we have
\begin{align*}
|A_1|&=O_P\left(\frac{p}{\hat\gamma_L} L\left(\frac{1}{\sqrt{T}} + \frac{1}{\sqrt{p}}\right)\right);\\
|A_2|&=O_P\left( \left(\frac{p}{\hat\gamma_L}\right)^2 \lambda_L^{-1/2} L^{5/2} \left( \frac{1}{T} + \frac{1}{p} \right)^{1/2} \right);\\
|A_3|&=O_P\left( \frac{p}{\hat\gamma_L} \lambda_L^{-1/2} L^2 T^{-1/2}\right).
\end{align*}
\begin{proof}
We have
\begin{align*}
\vert A_1 \vert &= \vert (\hat{f}_{T+1}^\prime - f_{T+1}^\prime H^\prime )\frac{1}{T+1}\hat{F}^\prime Y_{(-)} \vert \\
&\leq \sqrt{L} \Vert \hat{f}_{T+1}' - f_{T+1}^\prime H^\prime  \Vert \Vert Y_{(-)} \Vert / \sqrt{T+1}.
\end{align*}
By Lemma~\ref{l:Ybound} we have $\Vert Y_{(-)} \Vert / \sqrt{T+1}=O_P(1)$. By \eqref{L.4} the bound for $A_1$ follows.
As for $\vert A_2 \vert$
we note that
\begin{align*}
\vert A_2 \vert &= \vert \frac{1}{T+1}f_{T+1}'H' (\hat{F}^\prime - H^\prime F^\prime)Y \vert \\
&\leq \Vert f_{T+1} \Vert \Vert H \Vert \frac{1}{\sqrt{T+1}} \left\Vert \hat{F} - HF \right\Vert_F \Vert Y\Vert/ \sqrt{T+1}.
\end{align*}
One can easily see that $\Vert f_{T+1} \Vert = O_P(\sqrt{L})$. By \eqref{L.12} and \eqref{L.7} the bound for $\vert A_2 \vert$ follows immediately. 
\begin{align*}
\vert A_3 \vert &= \vert \frac{1}{T+1}f_{T+1}'H'\hat{f}_{T+1}Y_{T+1} \vert \\
&\leq \Vert f_{T+1} \Vert  \Vert H \Vert \Vert \hat{f}_{T+1} /\sqrt{T+1}\Vert \vert Y_{T+1}/\sqrt{T+1} \vert
\end{align*}
It is easily seen with Markov's inequality that $\vert Y_{T+1}/\sqrt{T+1} \vert= O_P(1/\sqrt{T})$. Note furthermore that $\Vert \hat{f}_{T+1} /\sqrt{T+1}\Vert \leq \Vert \hat{F} /\sqrt{T+1}\Vert_F = \sqrt{L}$. The result follows from \eqref{L.7} and previous considerations.
\end{proof}
\end{lemm}
\begin{lemm}\label{lemma:Bs} Under Assumptions \ref{ass:L}--\ref{a:regrerror} we have that
\begin{align*}
|B_1|&=O_P\left(\left(\frac{p}{\hat\gamma_L}\right)^2 \lambda_L^{-2} L^{7/2}  T^{-1/2}\right);\\
|B_2|&=O_P\left( \left(\frac{p}{\hat\gamma_L}\right)^2 \lambda_L^{-1} L^{3} T^{-1/2} \right);\\
\end{align*}
\begin{proof}
We see that
\begin{align*}
\vert B_1 \vert &= \left\vert f_{T+1}^\prime H^\prime H \left( \frac{1}{T+1} F^\prime F - I_L \right)b \right\vert \\
&\leq \Vert f_{T+1} \Vert  \Vert H \Vert^2 \left\Vert \frac{1}{T+1} F^\prime F - I_L \right\Vert_F \Vert b \Vert
\end{align*}
The bound for $\vert B_1 \vert$ then follows immediately from \eqref{L.7}, \eqref{L.11} and Lemma~\ref{l:bbound}.
\begin{align*}
\vert B_2 \vert &= \left\vert \frac{1}{T+1} f_{T+1}^\prime H^\prime H F^\prime \varepsilon \right\vert \\
&\leq \Vert f_{T+1} \Vert  \Vert H \Vert^2 \left\Vert F'\varepsilon / (T+1)\right\Vert.
\end{align*}
Using Markov's inequality it is easily seen that $\left\Vert F'\varepsilon / (T+1)\right\Vert=O_P(\sqrt{L/T})$.  The result follows again from \eqref{L.7}.
\end{proof}
\end{lemm}
\begin{lemm}\label{lemma:C} Under Assumptions \ref{a:noise}--\ref{a:pcs} we have that
$$|C|=O_P\left(\left(\frac{p}{\hat\gamma_L}\right)^4 \lambda_L^{-3} L^{11/2}  \left( \frac{1}{\sqrt{T}} + \frac{1}{\sqrt{p}} \right) \right)$$
\begin{proof}
The result follows immediately from the previous considerations and \eqref{L.8}.
\end{proof}
\end{lemm}
\begin{proof}[\textbf{Proof of Theorem~\ref{thm}}]
Following the Lemmas \ref{lemma:As}, \ref{lemma:Bs} and \ref{lemma:C} we see, given the assumptions of Theorem~\ref{thm}, that the dominant term in the prediction error is $\vert C \vert$. The result follows immediately. 
\end{proof}

\bibliography{FLRbib}

\end{document}